\documentclass[11pt]{article}
\usepackage{authblk}
\usepackage[toc,page]{appendix}
\usepackage[top=2cm, bottom=2cm, left=2cm, right=2cm]{geometry}
\usepackage{centernot}

\usepackage{color}
\usepackage{helvet}         
\usepackage{courier}        
\usepackage{type1cm}        

\usepackage{framed}
\usepackage{tikz}
\usepackage{makeidx}         
\usepackage{graphicx}        
\usepackage{multicol}        
\usepackage[bottom]{footmisc}

\usepackage{amsmath,mathrsfs,bbm}
\usepackage{amssymb}
\usepackage{bbold}
\usepackage{amsthm}
\usepackage{subcaption}
\usepackage{sidecap}
\usepackage{floatrow}
\usepackage{pdflscape}
\usepackage{pdflscape}
\usepackage{comment}
\usepackage[font=small]{caption}
\usepackage{enumitem}
\usepackage{esint}

\usepackage{scalerel}

\newtheorem{theorem}{Theorem}
\newtheorem{corollary}[theorem]{Corollary}
\newtheorem{lemma}[theorem]{Lemma}

\newtheorem{proposition}[theorem]{Proposition}

\theoremstyle{remark}

\newtheorem{remark}[theorem]{\bf Remark}

\numberwithin{theorem}{section}
\numberwithin{question}{section}
\numberwithin{figure}{section}
\numberwithin{equation}{section}

\begin{document}

\title{Power-law correction in the probability density function \\ of the critical Ising magnetization}
\bigskip{}
\author[1,2]{Federico Camia\footnote{Email: federico.camia@nyu.edu (corresponding author)}}
\author[3,4]{Omar El Dakkak\footnote{Email: omar.eldakkak@sorbonne.ae}}
\author[5]{Giovanni Peccati\footnote{Email: giovanni.peccati@uni.lu}}
\affil[1]{NYUAD, Abu Dhabi, UAE}
\affil[2]{Courant Institute of Mathematical Sciences, NYU, New York, USA}
\affil[3]{Sorbonne University Abu Dhabi, SAFIR, Abu Dhabi, UAE}
\affil[4]{MODAL'X, UMR CNRS 9023, Universit\'e Paris Nanterre, Paris, France}
\affil[5]{University of Luxembourg, Luxembourg}

\date{}

%

%

\maketitle

\begin{abstract}
    At the critical point, the probability density function of the Ising magnetization is believed to decay like $\exp{(-x^{\delta+1})}$, where $\delta$ is the Ising critical exponent that controls the decay to zero of the magnetization in a vanishing external field.
    In this paper, we discuss the presence of a power-law correction $x^{\frac{\delta-1}{2}}$, which has been debated in the physics literature.
    We argue that whether such a correction is present or not is related to the asymptotic behavior of a function that measures the extent to which the average magnetization of a finite system with an external field is influenced by the boundary conditions.
    Our discussion is informed by a mixture of heuristic calculations and rigorous results.
    Along the way, we review some recent results on the critical Ising model and prove properties of the average magnetization in two dimensions which are of independent interest.
\end{abstract}


\section{Background and motivation}

For decades, the Ising model has been one of the most studied and influential models of statistical mechanics.
It was introduced by Lenz~\cite{Lenz1920} in 1920 and its one-dimensional version was studied by Ising in his Ph.D. thesis \cite{Ising1924} and subsequent paper \cite{Ising1925}, but it was not until Peierls' and Onsager's famous investigations of the two-dimensional version that the model gained popularity.
In 1936 Peierls \cite{Peierls1936} proved that the two-dimensional version undergoes a phase transition, then in 1941 Kramers and Wannier~\cite{KW1941} located the critical temperature of the model defined on the square lattice and in 1944 Onsager~\cite{Onsager1944} derived its free energy. Since then, the two-dimensional Ising model has played a special role in the theory of critical phenomena: its phase transition has been extensively studied by both physicists and mathematicians and has become a prototypical example and a test case for developing ideas and techniques and for checking general hypotheses.

The Ising model attempts to capture ferromagnetism, one of the most interesting phenomena in solid state physics.
Ferromagnetism refers to the tendency of atomic spins to become spontaneously polarized in the same direction, generating a macroscopic magnetic field, as observed in some metals such as iron and nickel. However, this tendency is only present when the temperature is lower than a metal-dependent characteristic temperature, called \emph{Curie temperature} or \emph{critical temperature}. Above the Curie temperature the spins are oriented at random, producing no net magnetic field. Moreover, as the Curie temperature is approached from either side, the specific heat of the metal appears to diverge.

The Ising model is a crude but insightful attempt to reproduce the behavior described above. Its one-dimensional version fails to do so, as already realized by Ising \cite{Ising1924,Ising1925}, but the two-dimensional version does exhibit a phase transition, as shown by Peierls \cite{Peierls1936} and subsequently investigated by Kramers and Wannier \cite{KW1941}, Onsager \cite{Onsager1944} and many others.
In the most common version of the two-dimensional Ising model one associates a $\pm 1$ \emph{spin} variable to each vertex of a square grid and then assigns to each spin configuration a probability derived from a Gibbs distribution that favors the alignment of neighboring spins.
The appeal of the two-dimensional version of the model stems from its simplicity and the fact that it yields to an exact treatment, which has revealed a rich mathematical structure and has provided deep insights into the general theory of phase transitions.
In particular, various fundamental aspects of the theory of continuous phase transitions, or \emph{critical phenomena}, such as the \emph{scaling hypothesis} and the emergence of \emph{conformal invariance}, can be tested using the two-dimensional Ising model.

These aspects of the theory of critical phenomena can be heuristically understood in terms of the concept of \emph{correlation length}, roughly speaking the distance beyond which different parts of a system become uncorrelated.
More precisely, a physical system that is not close to a phase transition point will exhibit exponential decay of correlations, with a decay rate given by the inverse of the correlation length.
It can be observed experimentally that, when a physical system undergoes a continuous phase transition, the correlation length reaches the linear dimension of the experimental apparatus and fluctuations extend to all length scales (as in the fascinating phenomenon of critical opalescence \cite{Smo08,Ein10}).
Mathematically, this corresponds to a diverging correlation length, as one can check in the two-dimensional Ising model, and raises the following question: What happens to the correlations between, say, far-away spin variables in the Ising model?

The answer is that exponential decay is replaced by power-law decay, in a mathematical manifestation of the fact that the system no longer has a finite correlation length and that fluctuations extend to all length scales.
More generally, the lack of a finite correlation length implies that correlations and other functions of the system become homogeneous and exhibit a characteristic power-law behavior.
Interestingly, the exponents in these power laws appear to be universal across many different systems and models, and are largely independent of any microscopic details, depending only on global features such as symmetry and the number of spatial dimensions.

From a physics perspective, a natural and powerful way to study critical phenomena is to move from a discrete perspective to a continuum one and use a field-theoretic approach, focusing on large-scale properties and ignoring all microscopic details.
For the $d$-dimensional Ising model, this implies moving from a collection of random variables (the spin variables) associated to the vertices of a regular lattice (a lattice ``field'') to a ``continuum field'' defined on $\mathbb{R}^d$ and describing the fluctuations at scales much larger than that of the lattice spacing.

Mathematically, this idea can be implemented by taking a (\emph{continuum}) \emph{scaling limit}, which corresponds to sending the lattice spacing to zero while focusing on certain functions of the spin variables or on certain geometric structures. For example, focusing on interfaces leads to Schramm-Loewner evolution (SLE) curves (see~\cite{unpotutti14}). In this paper, we will follow the approach developed in~\cite{cn2009,cgn2015,cgn2016,cjn2020} and study the \emph{magnetization}, which can be expressed as a sum of spin variables over a given region of space.
The magnetization is of great interest to physicists because it has a clear physical meaning, representing the macroscopic magnetic field generated by a ferromagnetic material, and because it is the \emph{order parameter} of the paramagnetic-ferromagnetic phase transition. Moreover, since the spin variables are highly correlated at the Curie (critical) temperature, the magnetization is interesting from a probabilistic perspective because it provides an important example of a sum of strongly correlated random variables, leading to interesting non-central limit theorems (see~\cite{cn2009,Cam2012,cgn2015,cgn2016}).

The authors of~\cite{cgn2015} considered the critical Ising model on the square lattice with lattice spacing $a$, eventually sending $a$ to zero. Among other things, they showed that the magnetization in a unit square, rescaled by an appropriate power of $a$, has a nontrivial (random) limit in distribution, say $m$, as $a \to 0$.
In~\cite{cgn2016}, the same authors showed that the distribution of $m$ has an analytic density function and a tail
\begin{equation} \label{eq:leading_tail}
e^{-c x^{16}}    
\end{equation}
for some constant $c>0$ (see equation~\eqref{eq:expo-tail} below).

It is natural to ask if one can go beyond the tail behavior and, in particular, establish the asymptotic behavior of the density function.
The questions of identifying the asymptotic behavior for large $x$, its dependence on the size of the system and its role in the study of large deviations have a long history in physics and are still of current interest (see~\cite{tsypin94a, tsypin94b, Bru1995, hilwild95, porastella23, bdr2024, TS25}).
The problem of dependence on the size of the system pertains to the theory of finite-size scaling, which is an established and extensively studied branch of the theory of phase transitions, interesting in its own right, and essential to the understanding and effective use of Monte Carlo methods.
A full understanding of the asymptotic behavior is a natural endeavor from a mathematics perspective and is important to characterize the statistics of rare events, a problem that has been much studied in the last decades (see~\cite{bdr2024} and references therein).

For the Ising model in any dimension and for models in the Ising universality class, it has been argued~\cite{tsypin94a, tsypin94b, Bru1995, porastella23} that the density function of the distribution of the magnetization $m_L$ in a square of side length $L$ behaves like
\begin{align} \label{eq:prediction}
x^{\frac{\delta-1}{2}} e^{-C_Lx^{\delta+1}}
\end{align}
as $x \to \infty$, for some positive constant $C_L$, which depends on the size of the system.
However, the status of the leading order and that of the power-law correction in~\eqref{eq:prediction} are not the same.
The leading exponential part of~\eqref{eq:prediction}, which is consistent with~\eqref{eq:leading_tail} (and with~\eqref{eq:lim-gen-funct} below), seems to be unanimously accepted in the physics literature (see, e.g., \cite{bdr2024,TS25}). A derivation of this leading behavior is given in~\cite{MCWbook72} based on plausible assumptions. Some of those assumptions can probably be verified using the results and methods of~\cite{cgn2014} and~\cite{cgn2016}, but checking the validity of the arguments in~\cite{MCWbook72} is beyond the scope of this article, so here we will \emph{assume} the result, as explained in the next section.
On the contrary, doubts have been raised concerning the universality of the power-law prefactor. In particular, the authors of~\cite{bdr2024} argue that the prefactor is not always present because it can be hidden by an additional subleading term.
Here, we corroborate this conclusion and identify and study a function that could be responsible for hiding the power-law prefactor that appears in~\eqref{eq:prediction}.

For the two-dimensional Ising model, as already mentioned, the tail of the distribution of the critical magnetization was established in~\cite{cgn2016} and is given by \eqref{eq:leading_tail} (see also \cite{cgn2014} for a rigorous determination of the magnetization exponent $\delta=15$).
In this paper, under the assumption that the logarithm of the density function has a leading behavior analogous to \eqref{eq:leading_tail} at infinity (see \eqref{eq:density-decay}), we argue that the subleading behavior can be characterized in terms of the asymptotic behavior of an explicit function $\Delta_{L,\xi}(h)$ (see equation~\eqref{def:Delta} below), which depends on the boundary conditions $\xi$.
Roughly speaking, this function calculates the difference in average magnetization between a finite system with boundary condition $\xi$ and a portion of the infinite-volume system, when an external magnetic field is applied.
We propose the following criterion: if $\Delta_{L,\xi}(h)$ is integrable as $h\to\infty$, then the asymptotic behavior is given by~\eqref{eq:prediction} with $\delta=15$.

As an example, we discuss the case of periodic boundary conditions $\xi=p$, where we can exploit translation invariance. Assuming the existence of a continuum scaling limit with the same properties proved for other types of boundary conditions (see~\cite{cgn2015,cgn2016,cjn2020bis}), we prove that $\Delta_{L,p}(h)$ is integrable as $h\to\infty$.
A key step in the proof is Lemma~\ref{lemma-bounds}, stated and proved in the last section, which shows that, in the near-critical scaling limit with an external field (see Section~\ref{sec:defres}), away from the boundary, the influence of the boundary condition $\xi$ is small and decreases rapidly as the external field grows.

We carry out our analysis in two dimensions because of recent rigorous results for the critical Ising model on the square lattice. However, much of the analysis is independent of the dimension and the same function $\Delta_{L,\xi}$ should also determine the subleading behavior in higher dimensions, as we explain in Section~\ref{sec:discussion}.

In the next section, we introduce the Ising model on the square lattice and review some definitions and recent results.
In Section~\ref{sec:main-results}, we present some rigorous results about the average magnetization in finite domains, discuss the link between the mean of the magnetization and its moment-generating function, and introduce the function $\Delta_{L,\xi}$.
Section~\ref{sec:proof-main-statement} is concerned with the density of the magnetization and explains the role of the function $\Delta_{L,\xi}$.
In Section~\ref{sec:periodic}, we discuss the case of periodic boundary conditions.
In Section~\ref{sec:discussion}, we discuss the possibility of extending our arguments to higher dimensions.
The paper ends with a section containing the proofs of two key lemmas and a corollary, which are of interest in their own right.

\section{Definitions and some recent results} \label{sec:defres}

Take $a>0$, let $\Lambda_L := [-L/2,L/2]^d$ and let $\Lambda^a_L := a\mathbb{Z}^d \cap \Lambda_L$, interpreted as a graph whose vertices are the elements of $a\mathbb{Z}^d \cap [-L/2,L/2]^d$ and where two vertices are \emph{adjacent} if they are nearest neighbors.
The Ising model at inverse temperature $\beta>0$ with external field $H$ on $\Lambda^a_L$ with boundary condition $\xi$ is a probability measure $\mathbb{Q}^{\beta,H,a}_{L,\xi}$ on $\{-1,1\}^{\Lambda^a_L}$ defined as follows.
For any \emph{spin configuration} $\sigma = \{\sigma_x\}_{x\in\Lambda^a_L} \in \{-1,1\}^{\Lambda^a_L}$ and any $s=\{s_y\}_{y\in\partial\Lambda^a_L} \in \{-1,0,1\}^{\partial\Lambda^a_L}$, let 
\begin{equation}\label{def:energy}
E^{H,a}_{L,s}(\sigma):= -\sum_{\stackrel{x\sim y}{x,y \in \Lambda^a_L}} \sigma_x \sigma_y - H \, \sum_{x \in \Lambda^a_L} \sigma_x - \sum_{\stackrel{x\sim y}{x \in \Lambda^a_L, y \in \partial \Lambda^a_L}} \sigma_x s_y
\end{equation}
be the \emph{energy} of $\sigma$, where the first sum is over nearest neighbor pairs in $\Lambda^a_L$ and the last is over vertices in the \emph{outer boundary} $\partial \Lambda^a_L$ of $\Lambda^a_L$ (i.e., the set of vertices of $a\mathbb{Z}^d$ not in $\Lambda^a_L$ but adjacent to a vertex in $\Lambda^a_L$).

The probability measure $\mathbb{Q}^{\beta,H,a}_{L,\xi}$ is defined by
\begin{equation}\label{ising-dist-h}
\mathbb{Q}^{\beta,H,a}_{L,\xi}(\sigma) := \sum_{s \in \{-1,0,1\}^{\partial\Lambda^a_L}} P_{\xi}(s) \, \frac {e^{ -\beta\, E^{H,a}_{L,s}(\sigma)}}{Z^{\beta,H,a}_{L,s}} \, ,
\end{equation}
where $P_{\xi}$ is some prescribed distribution encoding the boundary condition $\xi$ and the \emph{partition function} $Z^{\beta,H,a}_{L,s}$ is defined as
\begin{align} \label{def:partition_function}
    Z^{\beta,H,a}_{L,s} := \sum_{\sigma \in \{-1,1\}^{\Lambda^a_L}} e^{ -\beta\, E^{H,a}_{L,s}(\sigma)} \, .
\end{align}
For example, for plus boundary conditions, $\xi=+$, we choose $P_+$ such that $P_+(s_y=1)=1$ for every $y \in \partial\Lambda^a_L$ and get
\begin{equation}\label{def:energy_+bc}
\mathbb{Q}^{\beta,H,a}_{L,+}(\sigma):=\frac{1}{Z^{\beta,H,a}_{L,+}} \, \exp{\left(\beta\sum_{\stackrel{x\sim y}{x,y \in \Lambda^a_L}} \sigma_x \sigma_y + \beta H \, \sum_{x \in \Lambda^a_L} \sigma_x + \beta\sum_{\stackrel{x\sim y}{x \in \Lambda^a_L, y \in \partial \Lambda^a_L}} \sigma_x \right)} \, ,
\end{equation}
where $Z^{\beta,H,a}_{L,+} \equiv Z^{\beta,H,a}_{L,\{ 1 \}^{\partial\Lambda^a_L}}$.
The case of minus boundary conditions, $\xi=-$, is analogous, while the choice $P_0(s_y=0)=1$ for every $y \in \partial\Lambda^a_L$ corresponds to \emph{free} boundary conditions.
We will also consider, in Section~\ref{sec:periodic}, the case of \emph{periodic} boundary conditions, denoted by $\xi=p$, in which corresponding vertices on opposite sides of $\Lambda^a_L$ are considered nearest neighbors. In this case there is no boundary and the last sum in the energy \eqref{def:energy} is not present. In two dimensions, this corresponds to placing the Ising model on a torus.

It is useful to introduce a partial order on spin configurations such that $\sigma \geq \sigma'$ if and only if $\sigma_x \geq \sigma'_x$ for each $x \in a\mathbb{Z}^2$ where $\sigma_x$ is defined.
Given two boundary conditions $\xi$ and $\xi'$, we say that $\xi$ \emph{dominates} $\xi'$ if $\mathbb{Q}^{H,a}_{L,\xi}$ stochastically dominates $\mathbb{Q}^{H,a}_{L,\xi'}$.
It follows from the FKG inequality that $\xi=+$ dominates all other boundary conditions, while $\xi=-$ is dominated by all other boundary conditions (see e.g. \cite{FVbook}).
It is also clear that $\mathbb{Q}^{H,a}_{L,\xi}$ stochastically dominates $\mathbb{Q}^{H',a}_{L,\xi}$ if $H>H'$.

It is a simple but important observation that the energy can be written as
\begin{equation}\label{eq:energy}
E^{H,a}_{L,s}(\sigma) = E^{0,a}_{L,s}(\sigma) - H \, M^a_L(\sigma) \, ,
\end{equation}
where
\begin{equation}\label{def:magnetization}
M^a_L(\sigma) := \sum_{x\in \Lambda^a_L} \sigma_x
\end{equation}
is the \emph{total magnetization} in $\Lambda^a_L$, which implies that
\begin{align} \label{eq:r-d-dwrivative}
    \mathbb{E}^{\beta,H,a}_{L,s}(\cdot)=\frac{\mathbb{E}^{\beta,0,a}_{L,s}\left(\cdot \, \mathrm{e}^{HM^a_L}\right)}{\mathbb{E}^{\beta,0,a}_{L,s}\left(\mathrm{e}^{HM^a_L}\right)} \, ,
\end{align}
where $\mathbb{E}^{\beta,H,a}_{L,s}$ denotes expectation with respect to $\mathbb{Q}^{\beta,H,a}_{L,s}$.

Except when otherwise stated, as in some parts of Sections \ref{sec:periodic} and \ref{sec:discussion}, we will consider the two-dimensional model. In particular, all rigorous results are stated for $d=2$.
Since the work of Kramers and Wannier~\cite{KW1941} and later Onsager~\cite{Onsager1944}, it is known that the Ising critical inverse temperature on the square lattice is
\begin{equation}\label{}
\beta_c = \frac 1 2 \ln(1+\sqrt{2})\,.
\end{equation}
In the rest of the paper, we fix $\beta=\beta_c$ and drop it from the notation, writing $\mathbb{Q}^{H,a}_{L,\xi}$ instead of $\mathbb{Q}^{\beta_c,H,a}_{L,\xi}$. We also restrict our attention to positive external fields: $H \geq 0$.
It is then well-known (see e.g. \cite{FVbook}) that the distribution $\mathbb{Q}^{H,a}_{L,\xi}$ has a unique infinite-volume (\emph{thermodynamic}) limit, as $L \to \infty$, independent of the choice of boundary conditions $\xi$, which we denote by $\mathbb{Q}^{H,a}$ (we observe that, trivially, each measure $\mathbb{Q}^{\beta,H,a}_{L,\xi}$ can be identified with a measure on the whole lattice $a\mathbb{Z}^2$).
We consider the random variable
\begin{equation}\label{def:rescaled_magnetization}
m^a_L := \sum_{x\in \Lambda^a_L} S_x = a^{15/8} \sum_{x\in \Lambda^a_L} S_x \, ,
\end{equation}
where $\{ S_x \}_{x \in \Lambda^a_L}$ is a collection of $\pm 1$ random variables with joint distribution $\mathbb{Q}^{H,a}_{L,\xi}$ or $\mathbb{Q}^{H,a}$, with expectations $\mathbb{E}^{H,a}_{L,\xi}$ and $\mathbb{E}^{H,a}$, respectively, and denote the distribution of $m^a_L$ by $\mathbb{P}^{H,a}_{L,\xi}$ or $\mathbb{P}^{H,a}$, respectively.

We will be interested in the (\emph{continuum}) \emph{scaling limit} $m_L$ of the rescaled magnetization $m^a_L$, as $a \to 0$.
It is shown in~\cite{cgn2015,cgn2016} that in two dimensions, at least for some natural choices of boundary conditions, including plus, minus and free, if the external field is chosen to be $H=ha^{15/8}$ with $h \geq 0$ (called the \emph{near-critical} or \emph{off-critical regime}), the distributions $\mathbb{P}^{H,a}_{L,\xi}$ and $\mathbb{P}^{H,a}$ have unique limits, which we denote by $\mathbb{P}^{h}_{L,\xi}$ and $\mathbb{P}^{h}$, respectively.
We will further denote by $\mathbb{E}^{h}_{L,\xi}$ and $\mathbb{E}^{h}$ the expectations with respect to $\mathbb{P}^{h}_{L,\xi}$ and $\mathbb{P}^{h}$, and will drop the superscript $h$ when $h=0$.
To summarize, we will use the following notation.
\begin{itemize}
    \item Infinite volume system with $h>0$: $m_L$ has distribution $\mathbb{P}^h$, with expectation $\mathbb{E}^h$.
    \item Infinite volume system with $h=0$: $m_L$ has distribution $\mathbb{P}$, with expectation $\mathbb{E}$.
    \item Finite system in $[-L/2,L/2]^2$ with boundary condition $\xi$ and $h>0$: $m_L$ has distribution $\mathbb{P}^h_{L,\xi}$, with expectation $\mathbb{E}^h_{L,\xi}$.
    \item Finite system in $[-L/2,L/2]^2$ with boundary condition $\xi$ and $h=0$: $m_L$ has distribution $\mathbb{P}_{L,\xi}$, with expectation $\mathbb{E}_{L,\xi}$.
\end{itemize}

It is believed and generally assumed (see, e.g., \cite{bdr2024} and references therein and~\cite{MCWbook72}) that $\mathbb{P}_{L,\xi}$ admits a density function $f_{L,\xi}$ such that
\begin{align} \label{eq:density-decay}
    \log f_{L,\xi}(x) \sim -C_L x^{16} \qquad \text{ as } x \to \infty \, ,
\end{align}
where $f(x) \sim g(x)$ as $x\to\infty$ means that $\lim_{x\to\infty}f(x)/g(x)=1$.
This is partially proved in~\cite{cgn2016}, where it is shown that, for various boundary conditions $\xi$, including plus, minus and free, $\mathbb{P}_{L,\xi}$ has an analytic density function and that
\begin{align} \label{eq:expo-tail}
    \log\mathbb{P}_{L,\xi}(m_L>x) \sim -K_L x^{16} \qquad \text{ as } x \to \infty \, ,
\end{align}
for some constant $K_L>0$, which is independent of the boundary condition $\xi$.
This can be written as 
\begin{align}
    \label{eq:theta1}
    \mathbb{P}_{L,\xi}(m_L>x)=e^{-K_Lx^{16}+\theta(x)},
\end{align}
where $\theta(x)=o(x^{16})$ as $x\to\infty$.
Showing that $\theta(x)$ is monotone for $x$ sufficiently large would suffice to prove \eqref{eq:density-decay}, but since we have no information on $\theta$, in this paper we \emph{assume} \eqref{eq:density-decay} and write
\begin{align} \label{eq:density}
    f_{L,\xi}(x)=\varphi_{L,\xi}(x) \, e^{-C_L x^{16}},
\end{align}
which implicitly defines the function
\begin{align} \label{def:density}
    \varphi_{L,\xi}(x):=e^{C_L x^{16}} \, f_{L,\xi}(x) \, .
\end{align}

\begin{remark} Inferring the asymptotic behavior of densities from tail estimates (or Laplace transform estimates) is a notoriously challenging task, often requiring regularity assumptions on the density, such as monotonicity. For an illustration of this fact, see e.g. \cite[p. 446]{Fellerbook} or \cite{seneta}.
\end{remark}

\medskip

We end this section with a brief discussion of the ansatz~\eqref{eq:density-decay}. To begin with, we note that the asymptotic behavior of the moment-generating function of $m_1$ is discussed in Proposition~2.2 of~\cite{cgn2016}, which shows that
\begin{align} \label{eq:lim-gen-funct}
    \lim_{h \to \infty} \frac{1}{h^{16/15}} \log\mathbb{E}_{1,\xi}\left(e^{h \, m_1}\right) = b
\end{align}
for some universal constant $b>0$, independent of the boundary conditions.

Using \eqref{eq:density}, the moment-generating function of $m_1$ can be written as
\begin{align}
    \mathbb{E}_{1,\xi}\left(e^{hm_1}\right) = \int_{-\infty}^{\infty} \varphi_{1,\xi}(x) e^{hx -C_1 x^{16}} \mathrm{d}x \, .
\end{align}
The function $g(x)=hx-C_1x^{16}$ has a unique maximum at $x_0=\left(\frac{h}{16 C_1}\right)^{1/{15}}$.
Applying Laplace's method to approximate the integral above leads to
\begin{align} \label{eq:Laplace-approx}
\begin{split}
    & \mathbb{E}_{1,\xi}\left(e^{hm_1}\right) \approx \int_{-\infty}^{\infty} \varphi_{1,\xi}(x) e^{g(x_0)+\frac{1}{2}g''(x_0)(x-x_0)^2} \mathrm{d}x \approx \varphi_{1,\xi}(x_0)e^{g(x_0)} \int_{-\infty}^{\infty} e^{-\frac{1}{2} \vert g''(x_0) \vert (x-x_0)^2} \mathrm{d}x \\
    & \qquad \qquad \quad = \sqrt{\frac{2\pi}{\vert g''(x_0) \vert}} \varphi_{1,\xi}(x_0) e^{g(x_0)} \propto \frac{\varphi_{1,\xi}\left(\left(\frac{h}{16 C_1}\right)^{1/{15}}\right)}{h^{7/{15}}} \exp{\left(\frac{15/16}{(16 C_1)^{1/{15}}}h^{16/{15}}\right)} \, ,
\end{split}
\end{align}
which is consistent with~\eqref{eq:lim-gen-funct}, suggesting moreover a link between $C_1$ and $b$.
In addition, we note the following intriguing fact. If $\varphi_{1,\xi}(x)$ grows like $x^7$, as predicted by \eqref{eq:prediction}, the powers of $h$ in front of the exponential in the right hand side of~\eqref{eq:Laplace-approx} cancel exactly.
In Section~\ref{sec:proof-main-statement}, we will use a similar approximation to extract more precise information about the function $\varphi_{L,\xi}$.


\section{Average magnetization and the moment-generating function} \label{sec:main-results}
Equations~\eqref{eq:expo-tail} and \eqref{eq:lim-gen-funct} from the previous section also hold for the infinite-volume measure, that is, with $\mathbb{P}_{L,\xi}$ and $\mathbb{E}_{L,\xi}$ replaced by $\mathbb{P}$ and $\mathbb{E}$, respectively.
In addition, in the infinite-volume case, it is easy to see that the following holds.
\begin{proposition} \label{prop:magnetization}
For all $L>0$ and $h\geq0$, we have that
\begin{align}
    \mathbb{E}^h(m_L) = B_L \, h^{1/15} \, ,
\end{align}
where $B_L = L^2 \, B_1$ and $B_1\in(0,\infty)$ is a universal constant.
\end{proposition}
\begin{proof}
If we denote by $\mathcal{N}(a)$ the number of vertices in $\Lambda^a_1$ and by $\sigma_0$ the spin at the origin, translation invariance implies that
\begin{align}
    \frac{\mathbb{E}^{H,a}(m_1^a)}{h^{1/{15}}} = \frac{a^{15/8} \mathcal{N}(a) \mathbb{E}^{H,a}(\sigma_0)}{h^{1/{15}}} = \frac{a^{2} \mathcal{N}(a) \mathbb{E}^{H,a}(\sigma_0)}{a^{1/8}h^{1/{15}}} = a^{2} \mathcal{N}(a) \frac{\mathbb{E}^{H,a}(\sigma_0)}{(a^{15/8}h)^{1/{15}}} = a^{2} \mathcal{N}(a) \frac{\mathbb{E}^{H,1}(\sigma_0)}{H^{1/{15}}} \, .
\end{align}
Since $\lim_{a \to 0}a^{2} \mathcal{N}(a)=1$, Theorem~4 of~\cite{cjn2020bis} implies that
\begin{align}
    \lim_{a \to 0}\frac{\mathbb{E}^{H,a}(m_1^a)}{h^{1/{15}}} = \lim_{H \to 0} \frac{\mathbb{E}^{H,1}(\sigma_0)}{H^{1/{15}}} = B_1 \, ,
\end{align}
where $B_1$ is the constant $B\in(0,\infty)$ in Theorem~4 of~\cite{cjn2020bis}.
The convergence in distribution of $m^a_1$ (see Theorem~1.4 of~\cite{cgn2016}) and moment-generating function bounds (see Proposition~3.5 of~\cite{cgn2015}) imply the convergence of $\mathbb{E}^{H,a}(m^a_1)$ to $\mathbb{E}^{h}(m_1)$, so we can conclude that $\mathbb{E}^{h}(m_1)=B_1h^{1/{15}}$.
Furthermore, the scaling properties of the magnetization (see Section~4.2 of~\cite{cjn2020}) imply that
\begin{align} \label{eq:B_L-scaling}
    \mathbb{E}^h(m_L) = L^{15/8} \, \mathbb{E}^{L^{15/8}h}(m_1) = L^{15/8} \, B_1 (L^{15/8} h)^{1/15} = B_1 \, L^2 \, h^{1/15} \;\;\; \forall h\geq0 \, ,
\end{align}
so that $\mathbb{E}^h(m_L) = B_L \, h^{1/15}$ with $B_L = L^2 \, B_1$.
\end{proof}

When $\mathbb{E}^h$ is replaced by $\mathbb{E}^h_{L,\xi}$, the simple proof above does not work because of the lack of translation invariance.
However, with the right ingredients, it is not difficult to obtain the following related results, whose proofs, except for that of Lemma~\ref{lemma_h-lim}, are postponed to Section~\ref{sec:proofs}.
\begin{lemma} \label{lemma_limit}
    For any boundary conditions $\xi$ 
    such that $\mathbb{P}^{ha^{15/8},a}_{L,\xi}$ has a (subsequential) weak limit $\mathbb{P}^{h}_{L,\xi}$, 
    we have that
\begin{align} \label{eq:limE(m_L)}
    \lim_{L\to\infty}\frac{1}{L^2}\mathbb{E}^h_{L,\xi}(m_{L}) = \frac{16}{15} \, b \, h^{1/15} \;\;\; \forall h\geq0 \, ,
\end{align}
where $b$ is the constant appearing in~\eqref{eq:lim-gen-funct}.
\end{lemma}

\begin{corollary} \label{corollary:bB}
The constants $B_1$ and $b$ of Proposition~\ref{prop:magnetization} and Lemma~\ref{lemma_limit} are related as follows:
\begin{align} \label{eq:bB}
    B_1 = \frac{16}{15} \, b \, .
\end{align}
\end{corollary}

\begin{lemma} \label{lemma_h-lim}
For any $L>0$ and any boundary conditions $\xi$ such that $\mathbb{P}^{ha^{15/8},a}_{L,\xi}$ has a (subsequential) weak limit $\mathbb{P}^{h}_{L,\xi}$,
    \begin{align} \label{eq:Eh_bounds}
    \mathbb{E}^h_{L,\xi}(m_L) \sim B_1 L^2 h^{1/{15}}, \, \text{ as } h\to\infty \, .
\end{align}
\end{lemma}

\begin{proof}
Assume that $\xi=+$ or $\xi=-$, then the scaling properties of the magnetization (see~\cite{cgn2016} and Section~4.2 of~\cite{cjn2020}) imply that, for every $\alpha>0$,
\begin{align} \label{eq:scaling-m_L}
    \mathbb{E}^h_{L,\xi}(m_L) = \alpha \, \mathbb{E}^{\alpha h}_{L/\alpha^{8/15},\xi}(m_{L/\alpha^{8/15}}) \, .
\end{align}
Taking $\alpha=1/h$, we get
\begin{align} \label{eq:L_h}
    \mathbb{E}^h_{L,\xi}(m_L) = \frac{1}{h} \, \mathbb{E}^{1}_{h^{8/15}L,\xi}(m_{h^{8/15}L}) \, .
\end{align}
Introducing the variable $L_h=h^{8/15}L$, we can write
\begin{align} \label{eq:LtoL_h}
    \frac{\mathbb{E}^h_{L,\xi}(m_L)}{L^2} = \frac{\mathbb{E}^{1}_{L_h,\xi}(m_{L_h})}{L_h^2} \, h^{1/15} \, ,
\end{align}
where, thanks to Lemma~\ref{lemma_limit},
\begin{align}
    \lim_{h \to \infty}\frac{\mathbb{E}^{1}_{L_h,\xi}(m_{L_h})}{L_h^2} = \frac{16}{15} \, b \, .
\end{align} 
This and~\eqref{eq:LtoL_h}, combined with~\eqref{eq:bB}, imply that, for $\xi=+$ and $\xi=-$,
\begin{align}
    \lim_{h\to\infty} \frac{\mathbb{E}^h_{L,\xi}(m_L)}{L^2 h^{1/{15}}} = \frac{16}{15} b = B_1 \, .
\end{align}
As an immediate consequence of the FKG inequality, the same limit is obtained with all boundary conditions $\xi$, which completes the proof.
\end{proof}

The average magnetization in an external field, studied above, can be used to express the moment-generating function, as follows. For all deterministic boundary conditions, including plus, minus and free, \eqref{eq:r-d-dwrivative} implies
\begin{align} \label{eq:ratio}
    \mathbb{E}^h_{L,\xi}(m_L)=\frac{\mathbb{E}_{L,\xi}\left(m_L\mathrm{e}^{hm_L}\right)}{\mathbb{E}_{L,\xi}\left(\mathrm{e}^{hm_L}\right)} = \frac{\mathrm{d}}{\mathrm{d}t}\log\mathbb{E}_{L,\xi}\left( e^{tm_L} \right) \Big\vert_{t=h} \, ,
\end{align}
so that
\begin{align} \label{eq:integral}
    \log\mathbb{E}_{L,\xi}\left(e^{h \, m_L}\right) = \int_0^h \mathbb{E}^t_{L,\xi}(m_L) \, \mathrm{d}t \, .
\end{align}

Introducing the function
\begin{align} \label{def:Delta}
\Delta_{L,\xi}(h):=\mathbb{E}^h_{L,\xi}(m_L)-\mathbb{E}^h(m_L) \, ,
\end{align}
which can be interpreted as the excess average magnetization in a square, in the presence of an external field, due to the boundary condition $\xi$, with the help of Proposition~\ref{prop:magnetization}, \eqref{eq:integral} can be written as
\begin{align} \label{eq:free-energy}
    \log\mathbb{E}_{L,\xi}\left( e^{hm_L} \right) = 
    \frac{15}{16} B_1 L^2 h^{16/15} + \int_0^h \Delta_{L,\xi}(t) \, \mathrm{d}t \, ,
\end{align}
which, combined with~\eqref{eq:bB}, gives
\begin{align} \label{eq:error-term}
    \mathbb{E}_{L,\xi}\left( e^{hm_L} \right) = \exp{\left(b L^2 h^{16/15} + \int_0^h \Delta_{L,\xi}(t) \, \mathrm{d}t \right)} \, .
\end{align}

We see that the integral of $\Delta_{L,\xi}$ tells us how much the moment-generating function deviates from $\exp{(bL^2h^{16/{15}})}$.
This motivates the following definition:
\begin{align} \label{def:funct-D_xi}
\mathcal{D}_{\xi}(L,h) := \int_0^{h} \Delta_{L,\xi}(t) \, \mathrm{d}t \, .
\end{align}
In the physics literature, this function is sometimes identified with the \emph{singular part of the finite-size free energy} (as in the seminal paper by Privman and Fisher \cite{pf1984}) or called the \emph{excess free energy}, especially in the literature on the Casimir effect (see, e.g., \cite{ddg2006}).\footnote{We thank Adam Ran\c{c}on for pointing this out to us.}

It is easy to see, using the same scaling properties of the magnetization (see Section~4.2 of~\cite{cjn2020}) leading to \eqref{eq:B_L-scaling}, that $\mathcal{D}_{\xi}(L,h)$ is really a function of a single variable:
\begin{align} \label{eq:single-variable}
\mathcal{D}_{\xi}(L,h) = \int_0^{h} L^{15/8}\Delta_{1,\xi}(L^{15/8}t) \, \mathrm{d}t = \int_0^{L^{15/8}h} \Delta_{\xi}(t) \, \mathrm{d}t = \mathcal{D}_{\xi}(1,L^{15/8}h) \, ,
\end{align}
where $\Delta_{\xi}(t):=\Delta_{1,\xi}(t)$.
Moreover, combined with the scaling properties of $m_L$ (see \cite{cgn2015}), \eqref{eq:error-term} and the definition of $\mathcal{D}_{\xi}(L,h)$, \eqref{eq:lim-gen-funct} immediately implies that
\begin{align} \label{eq:o}
    \lim_{h \to \infty}\frac{\mathcal{D}_{\xi}(L,h)}{h^{16/15}}=0 \, .
\end{align}

\section{Heuristic analysis of the density} \label{sec:proof-main-statement}
In this section, we will use the expression
\begin{align} \label{eq:expression}
    \mathbb{E}^h_{L,\xi}(m_L)=\frac{\mathbb{E}_{L,\xi}\left(m_L\mathrm{e}^{hm_L}\right)}{\mathbb{E}_{L,\xi}\left(\mathrm{e}^{hm_L}\right)}
\end{align}
to extract information about $\varphi_{L,\xi}$.
We will use a Laplace-type approximation, analogous to that leading to~\eqref{eq:Laplace-approx}, without studying the error, so the content of this section should be seen as a heuristic motivation for the introduction and study of the functions $\Delta_{L,\xi}$ and $\mathcal{D}_{\xi}$.

The behavior of the left-hand side of~\eqref{eq:expression} for large $h$ is given by Lemma~\ref{lemma_h-lim}, while the moment-generating function can be expressed using~\eqref{eq:error-term}, so we only need to study the remaining term, which, using \eqref{eq:density}, can be written as
\begin{align} \label{eq:explicit-form}
\mathbb{E}_{L,\xi}\left(m_L\mathrm{e}^{hm_L}\right)=\int_{-\infty}^{\infty} x \, \varphi_{L,\xi}(x) \, \exp\big({g_L(x)}\big) \, \mathrm{d}x \, ,
\end{align}
where
\begin{align} \label{def:g}
    g_L(x):=hx-C_Lx^{16} \, .
\end{align}
It is immediate to check that
\begin{align}
    x_0:=\left(\frac{h}{16C_L}\right)^{1/15}
\end{align}
is a critical point for $g_L$ and that $g_L''(x_0)<0$. Applying to the integral \eqref{eq:explicit-form} the same approximation used in~\eqref{eq:Laplace-approx} gives
\begin{align}
\begin{split} \label{eq:steepest}
   & \mathbb{E}_{L,\xi}\left(m_L\mathrm{e}^{hm_L}\right) \approx \sqrt{\frac{2\pi}{|g_L''(x_0)|}}\,x_0\,\varphi_{L,\xi}(x_0)\,\mathrm{e}^{g(x_0)} \\
   & \qquad = \frac{\tilde\gamma}{C_L^{1/10}} \frac{\varphi_{L,\xi}\left(\frac{h^{1/15}}{(16C_L)^{1/15}}\right)\exp\left(\frac{15/16}{(16C_L)^{1/15}}\,h^{16/15}\right)}{h^{7/15}} \, h^{1/15} \, ,
\end{split}
\end{align}
where
\begin{align}
    \tilde\gamma=16^{2/5}\,\sqrt{\frac{\pi}{120}} \, .
\end{align}
Combining \eqref{eq:ratio}, \eqref{eq:error-term} and \eqref{eq:steepest}, we obtain 
\begin{align} \label{eq:ilbestione}
\begin{split}
    & \mathbb{E}^h_{L,\xi}(m_L)\approx\frac{\tilde\gamma}{C_L^{1/10}} \, \frac{\varphi_{L,\xi}\left(\frac{h^{1/15}}{(16C_L)^{1/15}}\right)}{h^{7/15}} \, \exp\left(\left(\frac{15/16}{(16C_L)^{1/15}} - b L^2 - \frac{\mathcal{D}_{\xi}(L,h)}{h^{16/15}} \right) \, h^{16/15} \right) \, h^{1/15} \, .
\end{split}
\end{align}

Comparing~\eqref{eq:ilbestione} to~\eqref{eq:Eh_bounds} and using~\eqref{eq:o} shows that
\begin{align} \label{eq:bilanciamiento}
    \frac{15/16}{(16C_L)^{1/15}} = b L^2,
\end{align}
which lead to
\begin{align} \label{eq:C_L}
    C_L = \frac{15^{15}}{16^{16}} \frac{1}{b^{15}} \frac{1}{L^{30}} \, .
\end{align}
In fact, note that $\varphi_{L,\xi}$ cannot cancel the exponential in~\eqref{eq:ilbestione} as $h\to\infty$, because that would imply that $\varphi_{L,\xi}(x)$ behaves like $e^{-\tilde{C}_L \, x^{16}}$ as $x\to\infty$, which would simply amount to redefining the constant $C_L$ in~\eqref{eq:density}.
Therefore, \eqref{eq:ilbestione} becomes
\begin{align} \label{eq:labestiolina}
\begin{split}
    \mathbb{E}^h_{L,\xi}(m_L)\approx\gamma \, L^3 \, \frac{\varphi_{L,\xi}\left( B_1 L^2 h^{1/15} \right)}{e^{\mathcal{D}_{\xi}(L,h)} \, h^{7/15}} \, h^{1/15} \, ,
\end{split}
\end{align}
where
\begin{equation} \label{eq:gamma}
\gamma = \sqrt{\frac{2\pi}{15}} B_1^{3/2} \, .
\end{equation}

Comparing~\eqref{eq:labestiolina} to~\eqref{eq:Eh_bounds} suggests that
\begin{align}
 L\varphi_{L,\xi}\left( B_1 L^2 h^{1/15} \right) \text{ behaves like } e^{\mathcal{D}_{\xi}(L,h)} \, h^{7/15} \text{ for large } h \, ,   
\end{align}
which would imply that, whether or not $\varphi_{L,\xi}(x)$ behaves like $x^7$ for large $x$ depends on the asymptotic behavior of $\mathcal{D}_{\xi}(L,h)$ for large $h$. 
In particular, if 
\begin{align} \label{eq:integrability}
    \lim_{h\to\infty}\mathcal{D}_{\xi}(L,h)=\int_0^{\infty}\Delta_{L,\xi}(t)\mathrm{d}t
\end{align}
is finite, then \eqref{eq:labestiolina} suggests that $\varphi_{L,\xi}(x)$ behaves like $x^7$ as $x \to \infty$.
Motivated by this observation, in the next section, we will discuss the limit \eqref{eq:integrability} in the case of periodic boundary conditions.

\section{Periodic vs other boundary conditions} \label{sec:periodic}
In this section, we discuss the case of periodic boundary conditions on opposite sides of $\Lambda_L$.
Although scaling limit results for the Ising model are typically stated and proved for plus, minus or free boundary conditions \cite{cgn2015,cgn2016}, and not for periodic boundary conditions, similar results and proofs should apply to the latter case (see e.g. Remark~3.2 of~\cite{cgn2016}). Therefore, we assume the existence of a unique scaling limit for the magnetization for the two-dimensional critical Ising model with periodic boundary conditions, with an analytic density function and the same scaling properties as for other boundary conditions. We call this assumption \emph{hypothesis H}.

With periodic boundary conditions, the model is invariant under translations along the coordinate axes.
Translation invariance is a manifestation of the fact that imposing periodic boundary conditions means that there is no real boundary (in two dimensions, it corresponds to placing the model on a torus). Therefore, this type of boundary condition seems particularly favorable for ensuring that $\Delta_{L,\xi}(h)$ decays fast enough to be integrable, when $h\to\infty$. This is the content of the next theorem.

\begin{theorem} \label{thm:periodic}
Assume hypothesis H, then
\begin{align}
    \int_0^{\infty}\vert\Delta_{L,p}(t)\vert \, \mathrm{d}t<\infty \, .
\end{align}
\end{theorem}
\begin{proof}
    Translation invariance implies that $\mathbb{E}^h_{L,p}(m_{L})=4\mathbb{E}^h_{L,p}(m_{L/2})$ and $\mathbb{E}^h(m_{L})=4\mathbb{E}^h(m_{L/2})$. Combined with Lemma~\ref{lemma-bounds} below, this implies that, for all $h$ sufficiently large,
    \begin{align}
        \vert\Delta_{L,p}(h)\vert = \big\vert\mathbb{E}^h_{L,p}(m_L)-\mathbb{E}^h(m_L)\big\vert = 4\big\vert\mathbb{E}^h_{L,p}(m_{L/2})-\mathbb{E}^h(m_{L/2})\big\vert \leq 2\hat{B}_1 L^2 h^{1/15} e^{-C(1)Lh^{15/8}} \, ,
    \end{align}
    which concludes the proof.
\end{proof}

The key element in the proof of Theorem~\ref{thm:periodic} is the combination of translation invariance with Lemma~\ref{lemma-bounds}, which essentially says that, for all sufficiently large values of $h$, the effect of the boundary is exponentially small at distance $L/2$.
Lemma~\ref{lemma-bounds} is also valid when the boundary condition is plus, but in that case there is no translation invariance.

The proof of Lemma~\ref{lemma-bounds} is given in Section~\ref{sec:proofs}, but we point out that it is based on two-dimensional arguments, which break down in higher dimensions.
However, heuristic considerations (as well as the upper bounds in Lemma~\ref{lemma-bounds}, if one recalls that the correlation length in two dimensions is of order $h^{-8/15}$ \cite{cjn2020}) suggest that the effect of the boundary may be negligible (exponentially small) at any distance much larger than the correlation length, and that this may be true for all dimensions $d \geq 2$.
In that case, for periodic boundary conditions, some version of Theorem~\ref{thm:periodic} should probably be true for all dimensions $d>2$, and this should guarantee the validity of~\eqref{eq:prediction}, as we argued at the end of the previous section.
In~\cite{bdr2024}, substantial evidence is provided for the validity of~\eqref{eq:prediction} in three dimensions with periodic boundary conditions.

For other boundary conditions, how fast $\Delta_{L,\xi}(h)$ decays when $h\to\infty$ should depend on the effect of the boundary within a distance of the order of the correlation length. 
This seems a very subtle question, but one can at least conjecture that the influence of the boundary should increase with the number of dimensions, possibly slowing down the decay of $\Delta_{L,\xi}(h)$.
In~\cite{bdr2024}, it is argued that, for the three-dimensional Ising model with free boundary conditions, the prediction~\eqref{eq:prediction} is invalid due to an additional correction to the leading order, which hides the term $x^{\frac{\delta-1}{2}}$.

To gain an intuitive understanding of the influence of the boundary, one can think of two extreme cases: $d=1$, where the boundary consists of only two vertices, and a regular tree, where the boundary contains a number of vertices proportional to the ``volume.''
It is well known \cite{Runnels67,Eggarter74,MHZ74} that, in the case of the Ising model on the complete Cayley tree, the partition function contains contributions from both sites deep within the tree, and sites close to or on the boundary. The contribution from the latter is not negligible, even in the thermodynamic limit.
This model could be an interesting playground to study the effect of the boundary on the distribution of the magnetization in high dimensions.

\section{Higher dimensions} \label{sec:discussion}

An interesting aspect of the argument of Section~\ref{sec:proof-main-statement} showing the relation between a power-law correction and the sensitivity of a finite system to a change of boundary conditions is that it is essentially independent of the dimension.
The two-dimensional case is special because there are many rigorous results, such as those mentioned in Section~\ref{sec:defres} and used in Section~\ref{sec:proof-main-statement}.
However, similar results are expected to hold in higher dimensions with different critical exponents.
For example, one of the main tools in the analysis carried out in this paper, used repeatedly, is scale invariance, which is expected to be a feature of criticality in all dimensions $d \geq 2$.

In dimension $d \geq 3$, assuming the existence of the scaling limit of the magnetization (when appropriately rescaled), in view of~\eqref{eq:Eh_bounds}, it is natural to expect that\footnote{Note that $d=4$ corresponds to the upper critical dimension for the Ising model, so one should expect logarithmic corrections (see~\cite{wr1973}).}
\begin{align} \label{eq:d_lim}
    \lim_{h\to\infty}\frac{\mathbb{E}^h_{L,\xi}(m_L)}{L^d h^{1/\delta}}=B_1(d) \, ,
\end{align}
for some constant $B_1(d)$ depending on the dimension $d$ and with the appropriate critical exponent $\delta$ ($\delta \approx 4.79$ for $d=3$ and $\delta=3$, the mean-field value, for $d>4$; $d=4$ corresponds to the upper critical dimension, therefore one expects to have the mean-field critical exponent, $\delta=3$, with logarithmic corrections~\cite{wr1973,bbs2014}).
If, in addition, one could show that
\begin{align}
\mathbb{E}^h(m_L)=B_1(d) \, L^d \, h^{1/\delta} \, ,    
\end{align}
then, dropping $d$ from $B_1(d)$ for notational convenience, \eqref{eq:ratio}, \eqref{eq:integral}, \eqref{def:Delta}, \eqref{def:funct-D_xi} and \eqref{eq:d_lim} would lead to
\begin{align}
     \lim_{h\to\infty}\frac{\mathbb{E}_{L,\xi}\left(m_L e^{hm_L}\right)\exp{\left(-\frac{\delta}{1+\delta} B_1 L^d h^{(\delta+1)/\delta}-\mathcal{D}_{L,\xi}(h) \right)}}{L^d h^{1/\delta}} = B_1 \, .
\end{align}

Assuming an analytic density for the magnetization in $d$ dimensions whose leading behavior as $x\to\infty$ is of the form $e^{-C_L x^{\delta+1}}$, with a dimension-dependent $C_L=C_L(d)$, and replacing \eqref{def:g} with
\begin{align}
    g_L(x):=hx-C_Lx^{\delta+1} \, ,
\end{align}
the Laplace approximation would suggest that
\begin{align}
L^{d/2}\varphi_{L,\xi}\left(B_1\,L^d\,h^{1/\delta}\right) \text{ behaves like } e^{\mathcal{D}_{\xi}(L,h)} h^{(\delta-1)/(2\delta)}  \text{ for large } h \, .
\end{align}
If $\mathcal{D}_{\xi}(L,h)$ has a finite limit as $h\to\infty$, this would imply that $\varphi_{L,\xi}(x)$ behaves like $x^{\frac{\delta-1}{2}}$ as $x \to \infty$.
However, if $\mathcal{D}_{\xi}(L,h)$ does not have a finite limit as $h\to\infty$, its asymptotic behavior may hide the term $x^{\frac{\delta-1}{2}}$. This is presumably what happens in three dimensions with free boundary conditions \cite{bdr2024}.

At this point, it is useful to discuss the validity of the assumption that the magnetization has an analytic density whose leading behavior is of the form $e^{-C_L x^{\delta+1}}$, as $x\to\infty$. We already pointed out that one should expect logarithmic corrections when $d=4$, the Ising upper critical dimension. In dimensions $d>4$, the nature of the scaling limit of the critical Ising magnetization is known to depend on the boundary conditions: Gaussian for free boundary conditions~\cite{Aizenman1982,Frohlich1982}, non-Gaussian for periodic boundary conditions~\cite{LPS2025}.
This marked difference is due to the influence of the boundary conditions on the finite-size critical behavior in dimensions
$d>4$, which has been studied extensively in the mathematics and physics literature for Ising and related models (see, for instance, \cite{cjn2021,LPS2025} and references therein).

For all dimensions $d>4$, it is expected that $\delta=3$, the mean-field value, which rules out a leading behavior of the form $e^{-C_L x^{\delta+1}}$, as $x\to\infty$, in the case of free boundary conditions, when the limit is Gaussian.
On the other hand, a behavior of the type $e^{-C_L x^4}$, as $x\to\infty$, is not ruled out for the case of periodic boundary conditions and is, moreover, consistent with the density of the renormalized magnetization in the critical Curie-Weiss model (the Ising model on the complete graph), which is known to converge in distribution to a random variable with density $c_1e^{-c_2 x^4}$ \cite{SG1973}.
Interestingly, one of the main conclusions of~\cite{pap2006} is that the behavior of macroscopic bulk quantities in the high-dimensional Ising model with periodic boundary conditions “parallels more closely the complete graph paradigm." 
For example, a lower bound of order $n^{d/2}$, the order of the critical susceptibility in the Curie-Weiss model, was proved in~\cite{pap2006} for the critical susceptibility of a $d$-dimensional $n \times n$ torus with $d>4$. A matching upper bound was recently proved in~\cite{LPS2025}, showing that the susceptibility of the critical Ising model on a high-dimensional torus does indeed behave like that of the critical Curie-Weiss model (see also~\cite{cjn2022}).

Based on these considerations, it is natural to conjecture that, in dimensions $d>4$, with periodic boundary conditions, the magnetization has an analytic density whose leading behavior is of the form $e^{-C_L x^{\delta+1}}$, as $x\to\infty$.
This would imply that the considerations in this section are relevant to the case of periodic boundary conditions in $d>4$.

In three dimensions, very little is known rigorously, but a recent study~\cite{bdr2024} provides evidence for the conclusion that the probability density of the magnetization decays like $e^{-C_L x^{\delta}}$ as $x\to\infty$, regardless of the boundary conditions, with a power-law correction $x^{\frac{\delta-1}{2}}$ for some (e.g., periodic), but not other (e.g., free), boundary conditions.
This suggests that the considerations in this section are also relevant to the case of the critical Ising model in three dimensions.

\section{Proofs of two key lemmas and a corollary} \label{sec:proofs}

\begin{proof}[Proof of Lemma~\ref{lemma_limit}]
Let
\begin{align}
    \mu_{L,\xi}(h):=\frac{\mathbb{E}^h_{L,\xi}(m_L)}{L^2} \, .
\end{align}
It is easy to see, using the FKG inequality and corresponding stochastic domination, that
\begin{align}
    \frac{\mathbb{E}_{2^{k+1},+}(m_{2^{k+1}})}{2^{2k+2}} \leq \frac{\mathbb{E}_{2^{k},+}(m_{2^{k}})}{2^{2k}} \, .
\end{align}
Therefore, the sequence $\mu_{2^k,+}(h)$ is decreasing.
Since $\mathbb{E}^h_{L,+}(m_L) \geq \mathbb{E}(m_L) = 0$ by symmetry and the FKG inequality, the sequence $\mu_{2^k,+}(h)$ is bounded from below, which implies that it has a limit, as $k\to\infty$, which we denote by $\mu_{+}(h)$.
A similar argument shows that the sequence $\mu_{2^k,-}(h)$ is increasing.
Moreover, since $\mathbb{E}^h_{L,-}(m_L) \leq \mathbb{E}^h_{L,+}(m_L)$ by the FKG inequality, the sequence $\mu_{2^k,-}(h)$ is bounded from above, which implies that it has a limit, as $k\to\infty$, which we denote by $\mu_{-}(h)$.

Next, we derive an upper bound for $\mu_{L,+}(h)$.
Recalling that $H=ha^{15/8}$ and using equation~(2.3) of~\cite{cgn2014} (note that in~\cite{cgn2014} $a=1$, so $H=h$, and $F(H)=\mathbb{E}^{a,H}_{\ell,+}(\sum_{x\in\Lambda^a_{\ell}}S_x$), in our notation, where one can think of $L$ in~(2.3) of~\cite{cgn2014} as $\ell/a$), we have that, for all $h>0$, all $\ell>0$ and all $a$ sufficiently small,
\begin{align}
\begin{split} \label{eq:discrete-boud}
    & \mathbb{E}^{a,H}_{\ell,+}(m^a_{\ell}) = a^{15/8} \, \mathbb{E}^{a,H}_{\ell,+}\left(\sum_{x\in\Lambda^a_{\ell}}S_x \right) \leq a^{15/8} \, C \, \left[\left(\frac{\ell}{a}\right)^{15/8} + H \left(\frac{\ell}{a}\right)^{15/4}\right] = C \, \big(\ell^{15/8} + h \, \ell^{15/4}\big) \, ,
\end{split}
\end{align}
where $C>0$ is the constant in equation~(2.3) of~\cite{cgn2014}.
The convergence of $m^a_{\ell}$ to $m_{\ell}$ in distribution, as $a \to 0$, combined with bounds on the moment generating function (see Proposition~3.5 of~\cite{cgn2015} or Proposition~2.1 of~\cite{cgn2016}), implies that
\begin{align}
    \lim_{a \to 0}\mathbb{E}^{a,H}_{\ell,+}(m^a_{\ell}) = \mathbb{E}^{h}_{\ell,+}(m_{\ell}) \, .
\end{align}
Therefore, taking the limit as $a \to 0$ of~\eqref{eq:discrete-boud}, we have that, for all $h, \ell>0$,
\begin{align} \label{eq:l-bound}
    \mathbb{E}^{h}_{\ell,+}(m_{\ell}) \leq C \, \big(\ell^{15/8} + h \, \ell^{15/4}\big) \, .
\end{align}

Now assume that $L \geq h^{-8/15}$, so that $[-L/2,L/2]^2$ can be partitioned into equal squares of side length $\ell = s \, h^{-8/15}$ for some $1 \leq s < 2$.
For this choice of $\ell$, the FKG inequality
and~\eqref{eq:l-bound} imply that
\begin{align}
    \frac{\mathbb{E}^h_{L,+}(m_L)}{L^2} \leq \frac{1}{L^2} \, \left(\frac{L}{\ell}\right)^2\mathbb{E}^h_{\ell,+}(m_{\ell}) = \frac{\mathbb{E}^h_{\ell,+}(m_{\ell})}{\ell^2} \leq C \, (s^{-1/8}+s^{7/4}) \, h^{1/15} \, .
\end{align}
Since $1 \leq s < 2$, there is a universal positive constant $\hat{B}_1<\infty$ such that
\begin{align} \label{eq:upper-bound}
    0\leq\frac{\mathbb{E}^h_{L,+}(m_L)}{L^2} \leq \hat{B}_1 \, h^{1/15} \, .
\end{align}

Using the scaling properties of $m_L$ and the FKG inequality, for all $L \geq 1$, we also have that
\begin{align} \label{eq:double-bound}
    \mathbb{E}_{1,-}(m_1)\leq\frac{\mathbb{E}_{1,-}(m_1)}{L^{1/8}}\leq\frac{\mathbb{E}_{L,-}(m_L)}{L^2}\leq\frac{\mathbb{E}^h_{L,-}(m_L)}{L^2}\leq\frac{\mathbb{E}^h_{L,+}(m_L)}{L^2} \leq \hat{B}_1 \, h^{1/15} \, .
\end{align}
Therefore, if $\xi=+$ or $-$, using Theorem~2.7 of~\cite{cgn2016}, \eqref{eq:integral}, \eqref{eq:upper-bound}, \eqref{eq:double-bound} and the dominated convergence theorem, we have that
\begin{align}
    bh^{16/15}=\lim_{k\to\infty}\frac{1}{2^{2k}}\log\mathbb{E}_{2^k,\xi}(e^{hm_{2^k}})=\lim_{k\to\infty}\frac{1}{2^{2k}}\int_0^h\mathbb{E}^t_{2^k,\xi}(m_{2^k}) \, \mathrm{d}t = \int_0^h \mu_{\xi}(t) \, \mathrm{d}t \, ,
\end{align}
which implies that
\begin{align}
    \mu_{\xi}(h)=\frac{16}{15} b h^{1/15} \, .
\end{align}

In order to deal with generic sequences $L\to\infty$, where $L$ is not necessarily of the form $2^k$, we can follow the strategy used in the proof of Lemma~2.11 of~\cite{cgn2016}.
More precisely, we can divide $\Lambda_L$ into the square $\Lambda_{2^k}$, with the largest $k$ such that $\Lambda_{2^k}\subset\Lambda_L$, and the annulus $\Lambda_L\setminus\Lambda_{2^k}$ (see Figure~1 of~\cite{cgn2016}).
With an argument similar to that in the proof of Lemma~2.11 of~\cite{cgn2016} or to those above, it is now easy to see that, if we call $m_{ann}$ the magnetization in the annulus, for both $\xi=+$ and $\xi=-$,
\begin{align}
    \lim_{L\to\infty}\frac{\mathbb{E}^h_{L,\xi}(m_{ann})}{L^2}=0 \, .
\end{align}
Therefore, we can conclude that
\begin{align}
    \lim_{L\to\infty}\frac{1}{L^2}\log\mathbb{E}_{L,-}(m_L)=\lim_{L\to\infty}\frac{1}{L^2}\log\mathbb{E}_{L,+}(m_L)=\frac{16}{15} b h^{1/15} \, .
\end{align}
As an immediate consequence of the FKG inequality, the same limit is obtained with all boundary conditions $\xi$, which completes the proof.
\end{proof}

\begin{proof}[Proof of Corollary~\ref{corollary:bB}]
Consider the Ising model in $\Lambda^a_L$ with boundary conditions on $\partial\Lambda^a_L$ chosen according to the infinite-volume distribution $\mathbb{Q}^{H,a}$ with $H=ha^{15/8}$.
We denote this boundary condition by $\eta_h$.
It is clear that
\begin{align} \label{eq:restriction}
    \mathbb{Q}^{H,a}_{L,\eta_h} = \mathbb{Q}^{H,a} \big\vert_{\Lambda^a_L} \, ,
\end{align}
where $\mathbb{Q}^{H,a} \big\vert_{\Lambda^a_L}$ is the measure $\mathbb{Q}^{H,a}$ restricted to spin configurations in $\Lambda^a_L$.
This implies that the law of the magnetization $m^a_L$ induced by $\mathbb{Q}^{H,a}_{L,\eta_h}$ converges, as $a \to 0$, to $\mathbb{P}_{L,\eta_h}^{h}\equiv\mathbb{P}^{h}$. Therefore,
\begin{align}
    \mathbb{E}^{h}_{L,\eta_h}(m_L) = \mathbb{E}^{h}(m_L) = B_1 L^2 h^{1/15} \, .
\end{align}
Combined with Lemma~\ref{lemma_limit}, this gives the relation $B_1=\frac{16}{15}b$.
\end{proof}

\begin{lemma} \label{lemma-bounds}
For any $\ell \leq L/2$ and all $h>0$ sufficiently large,
\begin{align} \label{eq:small_difference_+}
    0 \leq \mathbb{E}^h_{L,+}(m_{\ell})-\mathbb{E}^h(m_{\ell}) \leq \hat{B}_1 \ell^2 h^{1/15} e^{-C(1)Lh^{8/15}} \, ,
\end{align}
where $\hat{B}_1>0$ and $C(1)>0$ are the constants appearing in~\eqref{eq:upper-bound} and in Proposition~3.3 of~\cite{cjn2020}, respectively.
Under the same conditions, if hypothesis H is satisfied, we also have
\begin{align} \label{eq:small_difference_p}
    \big\vert\mathbb{E}^h_{L,p}(m_{\ell})-\mathbb{E}^h(m_{\ell})\big\vert \leq 2\hat{B}_1 \ell^2 h^{1/15} e^{-C(1)Lh^{8/15}} \, .
\end{align}
\end{lemma}
\begin{proof}
The scaling properties of the magnetization (see Section~4.2 of~\cite{cjn2020}) imply that, for every $\alpha>0$,
\begin{align}
    \mathbb{E}^h_{L,\xi}(m_{\ell}) = \alpha \, \mathbb{E}^{\alpha h}_{L/\alpha^{8/15},\xi}(m_{\ell/\alpha^{8/15}}) \, .
\end{align}
In particular, taking $\alpha=1/h$ and introducing the variables  $L_h=h^{8/15}L$ and $\ell_h=h^{8/15}\ell$, we have
\begin{align} \label{eq:reduction}
    \mathbb{E}^h_{L,+}(m_{\ell}) = \frac{1}{h} \, \mathbb{E}^{1}_{L_h,+}(m_{\ell_h}) \, \text{ and } \, \mathbb{E}^h(m_{\ell}) = \frac{1}{h} \, \mathbb{E}^{1}(m_{\ell_h}) \, .
\end{align}

Now let $R(a,L_h)$ denote the event that a spin configuration in $(\Lambda_{L_h}\setminus\Lambda_{L_h/2}) \cap a\mathbb{Z}^2$ contains a circuit of $+1$ spins winding around $\Lambda_{L_h/2}$.
The Edwards-Sokal coupling shows that this event is implied by the event $G(a,L_h)$ described in Definition~3.2 of~\cite{cjn2020}.
Moreover, if we let $\bar{\mathbb{Q}}^{a^{15/8},a}$ denote the FK measure on $a\mathbb{Z}^2$ at the critical (inverse) temperature with external field $H=a^{15/8}$ (see Section~2.1 of~\cite{cjn2020} for precise definitions), according to Proposition~3.3 of~\cite{cjn2020}, we have
\begin{align} \label{eq:FKbound}
    \bar{\mathbb{Q}}^{a^{15/8},a}(G^c(a,L_h) \leq e^{-C(1)L_h} \, ,
\end{align}
uniformly in $a \to 0$, where $G^c(a,L_h)$ denotes the complement of $G(a,L_h)$ and $C(1)$ is a universal constant.
The FKG inequality implies that
\begin{align}
    \mathbb{E}^{a^{15/8},a}(m^a_{\ell_h}) \geq \mathbb{E}^{a^{15/8},a}_{L_h,+}(m^a_{\ell_h}) \mathbb{Q}^{a^{15/8},a}(R(a,L_h)) + \mathbb{E}^{a^{15/8},a}_{\ell_h,-}(m^a_{\ell_h}) \left(1-\mathbb{Q}^{a^{15/8},a}(R(a,L_h))\right) \, .
\end{align}
Using the Edwards-Sokal coupling (see \cite{cjn2020,cjn2020bis}) and \eqref{eq:FKbound} and letting $a \to 0$, we obtain
\begin{align}
    \mathbb{E}^{1}(m_{\ell_h}) \geq \mathbb{E}^{1}_{L_h,+}(m_{\ell_h}) \left(1-e^{-C(1)L_h}\right) + \mathbb{E}^{1}_{\ell_h,-}(m_{\ell_h}) \left(1-\mathbb{Q}^{1}(R(L_h))\right) \, ,
\end{align}
where $\mathbb{Q}^{1}(R(L_h)):=\limsup_{a \to 0}\mathbb{Q}^{a^{15/8},a}(R(a,L_h)) \leq 1$.
Lemma~\ref{lemma_limit} implies that $\mathbb{E}^{1}_{\ell_h,-}(m_{\ell_h})>0$ for $h$ sufficiently large, therefore the last inequality leads to
\begin{align}
    \mathbb{E}^{1}(m_{\ell_h}) \geq \mathbb{E}^{1}_{L_h,+}(m_{\ell_h}) \left(1-e^{-C(1)L_h}\right)
\end{align}
for all $h$ sufficiently large. Using this and \eqref{eq:reduction}, we conclude that
\begin{align}
\begin{split}
    \mathbb{E}^{h}(m_{\ell}) = \frac{1}{h}\mathbb{E}^{1}(m_{\ell_h}) \geq \frac{1}{h}\mathbb{E}^{1}_{L_h,+}(m_{\ell_h}) \left(1-e^{-C(1)L_h}\right) = \mathbb{E}^{h}_{L,+}(m_{\ell}) \left(1-e^{-C(1)Lh^{8/15}}\right) \, .
\end{split}
\end{align}
Since $\mathbb{E}^{h}_{L,+}(m_{\ell}) \geq \mathbb{E}^{h}(m_{\ell})$, using the FKG inequality and~\eqref{eq:upper-bound}, this leads to
\begin{align}
\begin{split} \label{eq:lowerboundinfvol}
    0 \leq \mathbb{E}_{L,+}^{h}(m_{\ell}) - \mathbb{E}^{h}(m_{\ell}) \leq \mathbb{E}^{h}_{L,+}(m_{\ell}) e^{-C(1)Lh^{8/15}} \leq \mathbb{E}^{h}_{\ell,+}(m_{\ell}) e^{-C(1)Lh^{8/15}} \leq \hat{B}_1 \ell^2 h^{1/15} e^{-C(1)Lh^{8/15}} \, ,
\end{split}
\end{align}
which proves the first part of the lemma.

Since Proposition~3.3 of~\cite{cjn2020} applies to all types of boundary conditions on the FK model used in the Edwards-Sokal coupling, the same argument shows that
\begin{align}
\begin{split} \label{eq:lowerboundperiodic}
    0 \leq \mathbb{E}_{L,+}^{h}(m_{\ell}) - \mathbb{E}_{L,p}^{h}(m_{\ell}) \leq \hat{B}_1 \ell^2 h^{1/15} e^{-C(1)Lh^{8/15}} \, .
\end{split}
\end{align}
Now, using the triangle inequality together with \eqref{eq:lowerboundinfvol} and \eqref{eq:lowerboundperiodic}, we obtain
\begin{align}
    \big\vert\mathbb{E}^h_{L,p}(m_{\ell})-\mathbb{E}^h(m_{\ell})\big\vert \leq \left(\mathbb{E}_{L,+}^{h}(m_{\ell}) - \mathbb{E}^{h}(m_{\ell})\right) + \left(\mathbb{E}_{L,+}^{h}(m_{\ell}) - \mathbb{E}_{L,p}^{h}(m_{\ell})\right) \leq 2\hat{B}_1 \ell^2 h^{1/15} e^{-C(1)Lh^{8/15}} \, ,
\end{align}
which concludes the proof.
\end{proof}

\bigskip

\noindent{\bf Acknowledgments.} This project was started during a visit of the first two authors to the University of Luxembourg, whose kind hospitality they gratefully acknowledge. The first author thanks Adam Ran\c{c}on for an interesting correspondence, for bringing references \cite{bdr2024,Bru1995,pf1984} to his attention, and for useful comments on a draft of the manuscript. All authors thank an anonymous referee for a careful reading of the manuscript and useful suggestions. F.~Camia is supported by NYUAD through a personal research grant. G.~Peccati was supported by the Luxembourg National Research Fund (Grant: 021/16236290/HDSA). F.~Camia thanks the Hausdorff Research Institute for Mathematics (HIM) for its hospitality in July 2025 during the Trimester Program ``Probabilistic methods in quantum field theory'' funded by the Deutsche Forschungsgemeinschaft (DFG, German Research Foundation) under Germany's Excellence Strategy - EXC-2047/1 - 390685813.\\

\noindent{\bf Statements and Declarations:}
F. C. was at the time of submission an Editor of MPAG. Otherwise, on behalf of all authors, the corresponding author states that there is no conflict of interest and data sharing is not applicable to this article as no datasets were generated or analyzed during the current study.

\end{document}